\documentclass[submission,copyright,creativecommons]{eptcs}
\providecommand{\event}{PROLE 2015} 
\usepackage{amssymb}
\usepackage{amsmath}
\usepackage{latexsym}
\usepackage{makeidx}
\usepackage{epsfig}
\usepackage{cancel}
\usepackage{url}
\usepackage{tikz}
\usetikzlibrary{calc} 						
\usetikzlibrary{shapes,arrows}				
\usetikzlibrary{shapes.multipart}			
\usetikzlibrary{positioning}
\usepackage[all,cmtip]{xy}					
\usepackage{tikz-qtree,tikz-qtree-compat}	

\newtheorem{definition}{Definition}
\newtheorem{theorem}{Theorem}
  
\newtheorem{example}{Example}  
\newtheorem{remark}{Remark}  
  
\newtheorem{proposition}{Proposition}    
\newtheorem{corollary}{Corollary}
\newenvironment{proof}{\hskip-\parindent{\sc
Proof}.\ \ }{\hfill$\Box$\vskip\partopsep \vskip\topsep}

\newcommand{\nat}{\Bbb{N}}

\renewcommand{\emptyset}{\varnothing}
\renewcommand{\phi}{\varphi}

\newcommand{\toppos}{{\mbox{\footnotesize$\Lambda$}}} 
\newcommand{\ol}[1]{\overline{#1}}  
\newcommand{\pr}[1]{\mbox{\tt #1}}   
\newcommand{\mt}[1]{\mbox{\sf #1}}

\def\defemb#1#2{\expandafter\def\csname #1\endcsname
                              {\relax\ifmmode #2\else\hbox{$#2$}\fi}}
\defemb{cA}{{\cal A}}
\defemb{cB}{{\cal B}}
\defemb{cC}{{\cal C}}
\defemb{cD}{{\cal D}}
\defemb{cE}{{\cal E}}
\defemb{cF}{{\cal F}}
\defemb{cG}{{\cal G}}
\defemb{cH}{{\cal H}}
\defemb{cI}{{\cal I}}
\defemb{cJ}{{\cal J}}
\defemb{cL}{{\cal L}}
\defemb{cM}{{\cal M}}
\defemb{cN}{{\cal N}}
\defemb{cO}{{\cal O}}
\defemb{cP}{{\cal P}}
\defemb{cQ}{{\cal Q}}
\defemb{cR}{{\cal R}}
\defemb{cS}{{\cal S}}
\defemb{cT}{{\cal T}}
\defemb{cU}{{\cal U}}
\defemb{cV}{{\cal V}}
\defemb{cX}{{\cal X}}
\defemb{cZ}{{\cal Z}}

\defemb{mP}{{\sf P}}

\long\def\comment#1{}

\makeatletter
\newenvironment{prog}{\vspace{0.7ex}\par
\setlength{\parindent}{0.7cm}
\obeylines\@vobeyspaces\tt}{\vspace{0.7ex}\noindent
}
\makeatother
\newcommand{\startprog}{\begin{prog}}
\newcommand{\stopprog}{\end{prog}\noindent}

\makeatletter
\newenvironment{smallprog}{\vspace{0.7ex}\par
\setlength{\parindent}{0.7cm}
\obeylines\@vobeyspaces\tt\small}{\vspace{0.7ex}\noindent
}
\makeatother
\newcommand{\fstartprog}{\begin{smallprog}}
\newcommand{\fstopprog}{\end{smallprog}\noindent}

\makeatletter
\newenvironment{nismallprog}{\vspace{0.7ex}\par
\setlength{\parindent}{0.0cm}
\obeylines\@vobeyspaces\tt\small}{\vspace{0.7ex}\noindent
}
\makeatother
\newcommand{\fnistartprog}{\begin{nismallprog}}
\newcommand{\fnistopprog}{\end{nismallprog}\noindent}

\newcommand{\Pos}{{{\cal P}os}}

\defemb{HNF}{{\sf HNF}}
\defemb{NF}{{\sf NF}}
\defemb{GNF}{{\sf GNF}}
\defemb{NFinf}{{\sf NF}^\omega}

\defemb{Sfunc}{{\sf S}}
\defemb{Sred}{{\sf red}}
\defemb{Shnf}{{\sf hnf}}
\defemb{Snf}{{\sf nf}}
\defemb{SnfIn}{{\sf nf_{in}}}
\defemb{Seval}{{\sf eval}}
\defemb{Sempty}{{\sf empty}}
\defemb{Sinfred}{{\omega{\sf-red}}}
\defemb{Sinfhnf}{{\omega{\sf-hnf}}}
\defemb{Sinfeval}{{\omega{\sf-eval}}}
\defemb{Sinfnf}{{\omega{\sf-nf}}}
\defemb{Sinfgred}{{\infty{\sf-red}}}
\defemb{Sinfghnf}{{\infty{\sf-hnf}}}
\defemb{Sinfgeval}{{\infty{\sf-eval}}}
\defemb{Sinfgnf}{{\infty{\sf-nf}}}
\defemb{SinfgSred}{{\infty{\sf-Sred}}}
\defemb{SinfgShnf}{{\infty{\sf-Shnf}}}
\defemb{SinfgSeval}{{\infty{\sf-Seval}}}
\defemb{SinfgSnf}{{\infty{\sf-Snf}}}
\defemb{SinfSred}{{\omega{\sf-Sred}}}
\defemb{SinfShnf}{{\omega{\sf-Shnf}}}
\defemb{SinfSeval}{{\omega{\sf-Seval}}}
\defemb{SinfSnf}{{\omega{\sf-Snf}}}
\defemb{StransSg}{{\infty{\sf-S\varphi}}}
\defemb{Stransg}{{\infty{\sf-\varphi}}}
\defemb{SinfSg}{{\omega{\sf-S\varphi}}}
\defemb{Sinfg}{{\omega{\sf-\varphi}}}

\newcommand{\csr}{\mbox{\sc csr\/}}

\newcommand{\OBJTwo}{{\sf OBJ2}}
\newcommand{\OBJThree}{{\sf OBJ3}}
\newcommand{\CafeOBJ}{{\sf CafeOBJ}}
\newcommand{\Maude}{{\sf Maude}}

\newcommand{\Haskell}{{\sf Haskell}}

\newcommand{\muterm}{\mbox{\sc mu-term\/}}

\newcommand{\AProVE}{{\sf AProVE}}

\newcommand{\Symbols}{{\cF}}

\newcommand{\CSymbols}{{\cC}}
\newcommand{\DSymbols}{{\cD}}

\newcommand{\Variables}{{\cX}}
\newcommand{\TermsOn}[2]{{\cT(#1,#2)}}
\newcommand{\GTermsOn}[1]{{\cT(#1)}}

\newcommand{\Terms}{{\TermsOn{\Symbols}{\Variables}}}

\newcommand{\GTerms}{{\cT(\Symbols)}}
\newcommand{\CTerms}{{\cT(\CSymbols,\Variables)}}
\newcommand{\GCTerms}{{\cT(\CSymbols)}}
\newcommand{\ITerms}{{\cT^\omega(\Symbols,\Variables)}}

\newcommand{\IGCTerms}{{\cT^\omega(\CSymbols)}}

\newcommand{\Var}{{\cal V}ar}

\newcommand{\Rmaps}[1]{{M_{#1}}}
\newcommand{\CRmaps}[1]{{\mbox{\it CM}_{#1}}}
\newcommand{\muCan}{{\mu^{can}_{\cal R}}}

\newcommand{\ul}[1]{\underline{#1}}

\newcommand{\hto}{\hookrightarrow}

\newcommand{\activationlazyrew}[1]{\stackrel{\sf A}{\to}}
\newcommand{\activationlazyrewp}[1]{\toPlusPosSub{\sf A}{}}
\newcommand{\activationlazyrews}[1]{{\toStarPosSub{\sf A}{}}}

\newcommand{\csrew}[1]{\hto_{#1}}

\newcommand{\csrewpos}[2]{\stackrel{#1}{\hto}_{#2}}

\newcommand{\Fadd}{\mt{add}}

\newcommand{\Falt}{\mt{alt}}
\newcommand{\Ffbcons}{\mt{f$_{b:}$}}
\newcommand{\Fcons}{\mt{cons}}

\newcommand{\FconsF}{\mt{consF}}

\newcommand{\FevenNs}{\mt{evenNs}}
\newcommand{\Ffa}{\mt{f$_a$}}
\newcommand{\Ffb}{\mt{f$_b$}}
\newcommand{\Ffrac}{\mt{frac}}
\newcommand{\FhalfPi}{\mt{halfPi}}

\newcommand{\Fincr}{\mt{incr}}
\newcommand{\FL}{\mt{L}}

\newcommand{\Fnats}{\mt{nats}}
\newcommand{\FoddNs}{\mt{oddNs}}

\newcommand{\Fp}{\mt{p}}
\newcommand{\Fprod}{\mt{prod}}

\newcommand{\FprodFrac}{\mt{prodFrac}}
\newcommand{\FprodOfFracs}{\mt{prodOfFracs}}

\newcommand{\FrepItems}{\mt{rep2}}
\newcommand{\FS}{\mt{S}}

\newcommand{\Ftake}{\mt{take}}

\newcommand{\Fzip}{\mt{zip}}

\newcommand{\Fz}{\mt{0}}
\newcommand{\Fone}{\mt{1}}
\newcommand{\Fs}{\mt{s}}

\newcommand{\Ftail}{\mt{tail}}

\newcommand{\Fnil}{\mt{nil}}

\newcommand{\Fa}{\mt{a}}

\newcommand{\Fb}{\mt{b}}

\newcommand{\Ff}{\mt{f}}

\newcommand{\bigfracn}[3]{
\begin{array}[b]{c}
\displaystyle #1 \\\hline\displaystyle #2 
\end{array}
\hbox to 0pt{\raisebox{0.7em}{{\tiny (#3)}}}
}

\tikzstyle{decision} = [diamond, draw, fill=yellow!20, text width=5em, text badly centered, minimum height=4em, inner sep=0pt, aspect=2]
\tikzstyle{block} = [rectangle, draw,fill=blue!20, text width=5em, text centered, minimum height=4em, rounded corners]
\tikzstyle{cloud} = [ellipse, draw,fill=red!20, text width=5em, text centered, minimum height=4em]
\tikzstyle{line} = [draw, -latex']
\tikzstyle{blockR} = [rectangle, draw, fill=red!20, text centered, minimum height=4em, rounded corners, minimum height=0.75cm]
\tikzstyle{blockB} = [rectangle, draw, fill=blue!20, text centered, minimum height=4em, rounded corners, minimum height=0.75cm]
\tikzstyle{blockG} = [rectangle, draw, fill=green!20, text centered, minimum height=4em, rounded corners, minimum height=0.75cm]
\tikzstyle{blockY} = [rectangle, draw, fill=yellow!20, text centered, minimum height=4em, rounded corners, minimum height=0.75cm]
\tikzstyle{blockW} = [rectangle, draw, fill=white!20, text centered, minimum height=4em, rounded corners, minimum height=0.75cm]
\tikzstyle{blockK} = [rectangle, draw, fill=gray!20, text centered, minimum height=4em, rounded corners, minimum height=0.75cm]
\tikzstyle{mblockB} = [rectangle, draw, fill=blue!20, text centered, minimum height=4em, double,rounded corners, minimum height=0.75cm]
\tikzstyle{mblockW} = [rectangle, draw, fill=white!20, text centered, minimum height=4em, double,rounded corners, minimum height=0.75cm]
\tikzstyle{mblockR} = [rectangle, draw, fill=red!20, text centered, minimum height=4em, double,rounded corners, minimum height=0.75cm]
\tikzstyle{mblockG} = [rectangle, draw, fill=green!20, text centered, minimum height=4em, double,rounded corners, minimum height=0.75cm]

\tikzstyle{circleW} = [circle, draw, fill=white!20, text centered, minimum height=4em, rounded corners, minimum height=0.75cm]
\tikzstyle{triangleW} = [isosceles triangle, draw, fill=white!20, text centered, shape border rotate = 90, isosceles triangle stretches]
\tikzstyle{triangleG} = [isosceles triangle, draw, fill=green!20, text centered, shape border rotate = 90, isosceles triangle stretches]
\tikzstyle{triangleR} = [isosceles triangle, draw, fill=red!20, text centered, shape border rotate = 90, isosceles triangle stretches]

\begin{document}
\title{Termination of canonical context-sensitive rewriting and productivity of rewrite systems\thanks{Partially 
supported by the EU (FEDER), Spanish 
MINECO TIN 2013-45732-C4-1-P,
and GV PROMETEOII/2015/013.} }

\author{Salvador Lucas
\institute{DSIC,  Universitat Polit\`ecnica de Val\`encia, Spain\\ \url{http://users.dsic.upv.es/~slucas/}}}

\def\titlerunning{Termination of canonical context-sensitive rewriting and productivity of rewrite systems}

\def\authorrunning{Salvador Lucas}

\maketitle

\begin{abstract} 
Termination of programs, i.e., the absence of infinite computations, 
ensures the existence of normal forms for \emph{all} initial expressions, thus providing an
essential ingredient for the definition of a \emph{normalization semantics} for functional programs.
In \emph{lazy} functional languages, though,  \emph{infinite data structures} 
are often delivered as the \emph{outcome} of computations. For instance, the list of all prime numbers 
can be returned as
a neverending \emph{stream} of numerical expressions or data structures. 
If such streams are allowed, 
requiring 
termination is hopeless.
In this setting, the notion of \emph{productivity}  
can be used to provide an account of computations with infinite data structures, as it
``\emph{captures the idea of computability, of progress of infinite-list programs}''
(B.A.\ Sijtsma, On the Productivity of Recursive List Definitions, \emph{ACM Transactions on Programming Languages and Systems} 11(4):633-649, 1989).
However, in the realm of \emph{Term Rewriting Systems}, which can be seen as (first-order, untyped,
unconditional) functional 
programs, termination of \emph{Context-Sensitive Rewriting} (\csr) has been 
showed \emph{equivalent} to productivity of rewrite systems through appropriate
transformations.
In this way, tools for proving termination of \csr\ can be used to prove productivity.
In term rewriting, \csr\ is the restriction of rewriting that arises when 
reductions are allowed on selected arguments of function symbols only.
In this paper we show that well-known results about the computational power of \csr\ are useful
to better understand the existing connections between productivity of rewrite systems and termination of \csr,
and also to obtain more powerful techniques to prove productivity of rewrite systems.

\end{abstract}


\section{Introduction}

The computation of \emph{normal forms} of initial expressions provides an appropriate computational
principle for the semantic description of functional programs by means of a \emph{normalization semantics}
where initial expressions are given an associated normal form, i.e., an expression that do not issue any computation.
However, lazy functional languages (like \Haskell\ \cite{HudPeyWad_RepFuncProgLangHaskell_SN92}) 
admit giving {\em 
infinite values} as the meaning of expressions. Infinite 
values are limits of converging infinite sequences of {\em partially defined} 
values which are more and more defined and only contain \emph{constructor symbols}.
An appropriate notion of \emph{progress} in lazy functional computations is given by the notion of 
\emph{productivity} \cite{Sijtsma_ProductivityOfRecursiveListDefinitions_TOPLAS89} which concerns
the progress in the computation of infinite values when normal forms cannot be obtained.

Term Rewriting Systems (TRSs \cite{BaaNip_TermRewAllThat_1998,Ohlebusch_AdvTopicsTermRew_2002,Terese_TermRewritingSystems_2003}) 
provide suitable abstractions for functional programs which are often useful to
investigate their computational properties.
We can see a term rewriting system as a first-order functional program without any kind of type information associated to
any expression, and where all rules in the program are unconditional rules $\ell\to r$ where $\ell$ is a term
$f(\ell_1,\ldots,\ell_k)$ for some function symbol $f$ and terms $\ell_1,\ldots,\ell_k$, and 
$r$ is a term whose variables already occur in $\ell$.
The following example illustrates the use of infinite data structures with term rewriting systems.

\begin{example}\label{ExWallisExtended}
The TRS $\cR$ in Figure \ref{FigExWallisExtended}
\cite[Example 1]{AlaGutLuc_CSDPs_IC10}
can be used to compute approximations to $\frac{\pi}{2}$ as
$\frac{\pi}{2}=lim_{n\to\infty}\frac{2}{1}\frac{2}{3}\frac{4}{3}\frac{4}{5}\cdots\frac{2n}{2n-1}\frac{2n}{2n+1}$
(Wallis' product).
 \begin{figure}[t]
{\small
   \begin{eqnarray}
     \FevenNs & \to & \Fcons(\Fz,\Fincr(\FoddNs))\label{ExWallisExtended_rule1}\\
     \FoddNs & \to & \Fincr(\FevenNs)\label{ExWallisExtended_rule2}\\
     \Fincr(\Fcons(x,xs)) & \to & \Fcons(\Fs(x),\Fincr(xs))\label{ExWallisExtended_rule3}\\
    \Ftake(\Fz,xs) & \to & \Fnil\label{ExWallisExtended_rule4}\\
    \Ftake(\Fs(n),\Fcons(x,xs)) & \to & \FconsF(x,\Ftake(n,xs))\label{ExWallisExtended_rule5}\\
    \Fzip(\Fnil,xs) & \to & \Fnil\label{ExWallisExtended_rule6}\\
    \Fzip(xs, \Fnil) & \to & \Fnil\label{ExWallisExtended_rule7}\\
    \Fzip(\Fcons(x,xs),\Fcons(y,ys)) & 	\to & \Fcons(\Ffrac(x,y),\Fzip(xs,ys))\label{ExWallisExtended_rule8}\\
    \Ftail(\Fcons(x,xs)) & \to & xs\label{ExWallisExtended_rule9}\\
    \FrepItems(\Fnil) & \to & \Fnil\label{ExWallisExtended_rule10}\\
    \FrepItems(\Fcons(x,xs)) & \to & \Fcons(x,\Fcons(x,\FrepItems(xs)))\label{ExWallisExtended_rule11}\\
    \Fz+ x & \to & x\label{ExWallisExtended_rule12}\\
    s(x) + y & \to & \Fs(x + y)\label{ExWallisExtended_rule13}\\
    \Fz \times y & \to & \Fz\label{ExWallisExtended_rule14}\\
     s(x) \times y & \to & y+(x\times y)\label{ExWallisExtended_rule15}\\
    \FprodFrac(\Ffrac(x,y),\Ffrac(z,t)) & \to & \Ffrac(x\times z,y\times t)\label{ExWallisExtended_rule16}\\
    \FprodOfFracs(\Fnil) & \to & \Ffrac(\Fs(\Fz),\Fs(\Fz))\label{ExWallisExtended_rule17}\\
    \FprodOfFracs(\FconsF(p,ps)) & \to & \FprodFrac(p,\FprodOfFracs(ps))\label{ExWallisExtended_rule18}\\
    \FhalfPi(n) & \to & \FprodOfFracs(\Ftake(n,\Fzip(\FrepItems(\Ftail(\FevenNs)),\Ftail(\FrepItems(\FoddNs)))))\label{ExWallisExtended_rule19}
 \end{eqnarray}}
  \caption{Computing Wallis' approximation to $\frac{\pi}{2}$}\label{FigExWallisExtended}
 \end{figure}
\noindent 
In $\cR$, symbols $\Fz$ and $\Fs$ implement Peano's representation of natural numbers;
we also have the usual arithmetic operations $\Fadd$ition and $\Fprod$uct.
Symbols $\Fcons$ and $\Fnil$ are \emph{list constructors} to build (possibly
infinite) lists of natural numbers like
$\FevenNs$ (the infinite list of even numbers) and $\FoddNs$ (the
infinite list of odd numbers), which are defined by mutual recursion with rules (\ref{ExWallisExtended_rule1})
and (\ref{ExWallisExtended_rule2}).
Function $\Fincr$ increases the elements of a list in one unit through the 
application of $\Fs$ (rule (\ref{ExWallisExtended_rule3})).
Function $\Fzip$ merges a pair of lists into a list of fractions (rules (\ref{ExWallisExtended_rule6}) to (\ref{ExWallisExtended_rule8})), and
$\Ftail$ returns the elements of a list after removing the first one (rule (\ref{ExWallisExtended_rule9})).
Function $\Ftake$ (defined by rules (\ref{ExWallisExtended_rule4}) and (\ref{ExWallisExtended_rule5})) 
is used to obtain the components of  
a finite approximation to 
$\frac{\pi}{2}$ which we multiply with $\FprodOfFracs$, which calls the usual addition and product of
natural numbers defined by rules (\ref{ExWallisExtended_rule12}) to (\ref{ExWallisExtended_rule15}).
The explicit use of 
$\FconsF$ to build \emph{finite} lists of
fractions of natural numbers by means of $\Ftake$  ensures that the product of their elements 
computed by $\FprodOfFracs$ is well-defined.
A call $\FhalfPi(\Fs^n(\Fz))$ for some $n>0$ returns the desired approximation whose computation
is launched by rule (\ref{ExWallisExtended_rule19}).

Note that $\cR$ is \emph{nonterminating}. For instance we have the following infinite rewrite sequence:
\begin{eqnarray}
\ul{\FevenNs} \to \Fcons(\Fz,\Fincr(\ul{\FoddNs})) \to \Fcons(\Fz,\Fincr(\Fincr(\ul{\FevenNs})))\to\cdots\to\cdots\label{InfSequence}
\end{eqnarray}
\end{example}
\emph{Context-sensitive rewriting} (\csr\
\cite{Lucas_CScompFuncFunLogProg_JFLP98,Lucas_CSRewStrat_IC02}) is 
a restriction of rewriting which 
imposes fixed, \emph{syntactic} restrictions on reductions by means of 
a {\em replacement map} $\mu$ that, for each $k$-ary symbol
$f$, discriminates the argument positions
$i\in\mu(f)\subseteq\{1,\ldots,k\}$ which  {\em can} be rewritten and 
forbids them if $i\not\in\mu(f)$. 
These restrictions are raised to arbitrary subterms of terms in the obvious way. 
With \csr\ we can achieve a \emph{terminating behaviour} for TRSs $\cR$ which (as 
in Example \ref{ExWallisExtended}) are not
terminating in the unrestricted case.

\begin{example}
Let the replacement map 
$\mu$ be given by:
\[\mu(\Fcons) =\emptyset \text{ and } \mu(f) = \{1,\ldots, ar(f)\} \text{ for all } f\in\Symbols-\{\Fcons\}\]
That is, $\mu$ disallows rewriting on  the
   arguments of the list constructor $\Fcons$ (due to $\mu(\Fcons)=\emptyset$). This makes
  a kind of \emph{lazy evaluation} of lists possible.
  For instance, the rewrite sequence (\ref{InfSequence}) above is \emph{not} possible with \csr. The \emph{second}
  step is disallowed because the replacement is issued on the \emph{second} argument of $\Fcons$ and $2\notin\mu(\Fcons)$, i.e.,
 \begin{eqnarray}
\Fcons(\Fz,\Fincr(\ul{\FoddNs})) \not\hto_\mu \Fcons(\Fz,\Fincr(\Fincr(\FevenNs)))\nonumber
\end{eqnarray}
where we write $\hto_\mu$ to emphasize that the rewriting step is issued using \csr\ under the replacement map $\mu$.
This makes the infinite sequence impossible.
Termination 
of \csr\ for the TRS $\cR$ and $\mu$ in Example \ref{ExWallisExtended} can be automatically
proved with the termination tool \muterm\ \cite{AlaGutLucNav_ProvingTerminationPropertiesWithMUTERM_AMAST10}.
\end{example}
A number of programming languages like 
\CafeOBJ, \cite{FutNak_AnOvCafeSpecEnv_ICFEM97}, 
 \OBJTwo, \cite{FutGogJouaMes_PrincOfOBJ2_POPL85}, 
\OBJThree, \cite{GogWinMesFutJou_IntrodOBJ_2000}, 
and \Maude\ \cite{ClavelEtAl_MaudeBook_2007} 
admit the {\em explicit} 
specification of replacement restrictions under the so-called {\em local strategies}, which are sequences of 
argument indices associated to each symbol in the program.

Restrictions of rewriting may turn normal forms of some terms \emph{unreachable}, leading to {\em incomplete}
computations.
Sufficient conditions ensuring that 
context-sensitive computations stop yielding head-normal forms,
values or even normal forms have been investigated in \cite{Lucas_CSCompConfProg_PLILP96,%
Lucas_NeedRedCSRR_ALP97,%
Lucas_TransfEfEvFP_PLILP97,%
Lucas_CScompFuncFunLogProg_JFLP98,%
Lucas_CSRewStrat_IC02}.

The notion of \emph{productivity} in term rewriting has to do with the ability of TRSs to compute 
possibly infinite \emph{values} rather than arbitrary normal forms (as discussed in 
\cite{DershKapPlais_RewRewRewRewRew_TCS91,%
KennaKlopSleepVries_TransRedOrtTRS_IC95}, for instance).
In \csr, early results showed that, for left-linear TRSs $\cR$, if the replacement map $\mu$ is
made \emph{compatible} with the left-hand sides $\ell$ of the rules $\ell\to r$ of $\cR$,
then \csr\ has two properties which are specifically relevant for the purpose of this paper: 
\begin{enumerate}
\item\label{ItemMuNormalFormsAreHeadNormalForms} 
every $\mu$-normal form (i.e., a term $t$ where no further
rewritings are allowed with \csr\ under $\mu$) is a \emph{head-normal form} (i.e., a term that does not
rewrite into a redex) \cite[Theorem 8]{Lucas_CScompFuncFunLogProg_JFLP98},
\item\label{ItemConstructorHNFsAreMuReachable}
every term that rewrites into a \emph{constructor head-normal form} can be rewritten \emph{with \csr} 
into a constructor head-normal form with the same head symbol \cite[Theorem 9]{Lucas_CScompFuncFunLogProg_JFLP98}.
\end{enumerate}
The aforementioned \emph{compatibility} of the replacement map $\mu$ with the left-hand sides of 
the rules  (which is then called a \emph{canonical} replacement map) 
just ensures that the positions of nonvariable symbols in $\ell$ are always \emph{reducible}
under $\mu$.
For instance, $\mu$ in Example 
\ref{ExWallisExtended}
is a canonical replacement map for $\cR$  in the example.
See also \cite{Lucas_CompletenessOfContextSensitiveRewriting_IPL15} where the role of the canonical replacement in connection with the algebraic semantics of 
computations with \csr, as defined in \cite{HenMes_CompletenessOfCSOSSpec_RTA07} and also \cite{NakOgaFut_ReducOpSymbolsAndBehavSpec_JSC10}, has been investigated.

In the following, we show that the facts (\ref{ItemMuNormalFormsAreHeadNormalForms}) 
and (\ref{ItemConstructorHNFsAreMuReachable}) \emph{suffice} to prove that termination of \csr\ 
is a \emph{sufficient} condition for productivity (see Theorem \ref{TheoCanonicalMuTermAndProductivityOfTreeSpec}
below).
As mentioned before, the connection between termination of \csr\ and productivity is not new.
In particular, Zantema and Raffelsieper proved that termination of \csr\ is a sufficient condition
for productivity \cite{ZanRaf_ProvingProdInInfDataStructures_RTA10}, 
and then Endrullis and Hendriks proved that, in fact, and provided that some appropriate 
transformations are used, it is also \emph{necessary}, i.e., termination of \csr\ \emph{characterizes} 
productivity \cite{EndHen_LazyProductivityViaTermination_TCS11}.
\begin{example}\label{Ex6_8_EH11}
The following TRS $\cR$ 
can be used to define \emph{ordinal numbers} \cite[Example 6.8]{EndHen_LazyProductivityViaTermination_TCS11}: 
\[\begin{array}{rcl@{\hspace{1cm}}rcl}
x+\Fz & \to & x & x\times \Fz & \to & \Fz\\
x+\FS(y) & \to & \FS(x+y) & x\times \FS(y) & \to & (x\times y)+x\\
x+\FL(\sigma)) & \to & \FL(x+_L\sigma) & x\times\FL(\sigma) & \to & \FL(x\times_L\sigma)\\
x+_L(y:\sigma) & \to & (x+y):(x+_L\sigma) & x\times_L(y:\sigma) & \to & (x\times y):(x\times_L\sigma)\\
\Fnats(x) & \to & x:\Fnats(\FS(x)) & \omega & \to & \FL(\Fnats(\Fz))
\end{array}
\]
Here, $\Fz$ and $\FS$ are the usual constructors for natural numbers in Peano's notation;
a stream of ordinals can be obtained by means of the list constructor `$:$'  that combines an ordinal
and a stream of ordinals to obtain a new stream of ordinals; 
finally, $\FL$ represents a \emph{limit ordinal} defined by means of a stream of ordinals. 
For instance $\omega$ is given as the limit $\FL(\Fnats(\Fz))$ of $\Fnats(\Fz)$, the stream that contains all natural numbers.
Finally, $+$ and $\times$ are intended to, respectively, \emph{add} and \emph{multiply} ordinal numbers; 
symbol $+_L$ is an auxiliary operator that adds an ordinal number $x$ to a stream or ordinals 
by adding $x$ to each component of the stream using $+$.
Operation $\times_L$ performs a similar task with $\times$.
Endrullis and Hendriks use a transformation which introduces the replacement map
$\mu(+)=\mu(+_L)=\mu(\times)=\mu(\times_L)=\{2\}$, 
$\mu(\FS)=\{1\}$, and 
$\mu(\FL)=\mu(:)=\mu(\Fnats)=\emptyset$,
and also adds some rules to prove $\cR$ \emph{productive}.
\end{example}
So, what is our contribution? First, we show that the ability of \csr\ to prove productivity 
is a consequence of essential properties of \csr, like (1) and (2) above.
This theoretical clarification is valuable and useful for further developments in the field and, as far
as we know, has not been addressed before.
From a practical point of view, we are able to improve Zantema and Raffelsieper's criterion 
that uses unnecessarily `permisive' replacement maps which can fail to conclude 
productivity as termination of \csr\ in many cases.
For instance, we can prove productivity of $\cR$ in Example \ref{Ex6_8_EH11} as termination of \csr\ for the
replacement map $\mu$ in the example. Furthermore, we can do it automatically by using existing tools like
\AProVE\ \cite{GieSchThi_AproveForDPFramework_IJCAR06} or \muterm.
In contrast, with the replacement map $\mu'$ that would be obtained according to \cite{ZanRaf_ProvingProdInInfDataStructures_RTA10}, 
$\cR$ is \emph{not} terminating for \csr;
thus, productivity cannot be proved by using Zantema and Raffelsieper's technique.
We are also able to improve the treatment in  \cite{EndHen_LazyProductivityViaTermination_TCS11} because
they need to apply a transformation to $\cR$ that we do not need to use.
In fact, we were able to deal with all examples of productivity in those papers by using our main result
together with the aforementioned termination tools to obtain automatic proofs.
Our result, though, does \emph{not} provide a characterization of productivity, as we show by means of an
example.

However, our results apply to \emph{left-linear} TRSs, whereas 
\cite{EndHen_LazyProductivityViaTermination_TCS11,ZanRaf_ProvingProdInInfDataStructures_RTA10}
deal with \emph{orthogonal} (constructor-based) TRSs only.
Actually, we also supersede the main result of \cite{Raffelsieper_ProductivityOfNonOrthogonalTRSs_EPTCS12}
which applies to non-orthogonal TRSs which are still left-linear.
This is also interesting to understand the role of \csr\ in proofs of productivity.
Actually, the results in the literature about completeness of \csr\ to obtain head-normal forms and values concern
left-linear TRSs and canonical replacement maps only.
The additional restrictions that are usually imposed on TRSs to achieve productivity as termination of \csr\ (in
particular, \emph{exhaustive} patterns in the left-hand sides) have to do with the notion of productivity rather than
with \csr\ itself.

After some preliminaries in Section \ref{SecPreliminaries}, Section \ref{SecCSR} introduces the notions
about \csr\ that we need for the development of our results on productivity via termination of \csr\ in 
Section \ref{SecInfNormalizationProductivity}.
Section \ref{SecRelatedWork} compares with related work and Section \ref{SecConclusions} concludes.

\section{Preliminaries}\label{SecPreliminaries}

This section collects a number of definitions and notations about term rewriting \cite{BaaNip_TermRewAllThat_1998,Terese_TermRewritingSystems_2003}.
Throughout the paper, $\Variables$ denotes a 
countable set of variables and $\Symbols$ denotes
a signature, i.e., a set of function symbols
$\{\pr{f}, \pr{g}, \ldots \}$, each having a fixed arity given by a 
mapping $ar:\Symbols\rightarrow \nat$. The set of
terms built from $\Symbols$ and $\Variables$ is $\Terms$.
Given a (set of) term(s) $t\in\Terms$ (resp.\ $T\subseteq\Terms$), we write
$\Symbols(t)$ (resp.\ $\Symbols(T)$) to denote the subset of symbols in $\Symbols$
occurring in $t$ (resp. $T$).
A term is said to be linear if 
it has no multiple occurrences of a single variable.
Terms are viewed as labelled trees in the usual way. 
Positions $p,q,\ldots$
are represented by chains of positive natural numbers used to address subterms
of $t$. Given positions $p,q$, we denote its concatenation as $p.q$. 
Positions are ordered by the standard prefix ordering $\leq$. Given a 
set of positions $P$, $minimal_\leq(P)$ is the set of minimal 
positions of $P$ w.r.t.\ $\leq$. 
If $p$ is a position, and $Q$ is a set of positions,
$p.Q=\{p.q~|~q\in Q\}$. We denote the empty chain by $\toppos$. 
The set of positions of a term $t$ is $\Pos(t)$. 
Positions of non-variable symbols in $t$ are denoted as $\Pos_\Symbols(t)$, and 
$\Pos_\Variables(t)$ are the positions of variables.
The subterm at position $p$ of $t$ is denoted as 
$t|_p$ and  $t[s]_p$ is the
term $t$ with the subterm at position $p$ replaced by $s$.
The symbol labelling the root of $t$ is denoted as $root(t)$. 
Given terms $t$ and $s$, $\Pos_s(t)$ denotes 
the set of positions of $s$ in $t$, i.e., $\Pos_s(t)=\{p\in 
\Pos(t)\mid t|_p=s\}$. 
A substitution is a mapping $\sigma:\Variables\to\Terms$ which is homomorphically
extended to a mapping $\sigma:\Terms\to\Terms$ which, by abuse, we denote using the same symbol
$\sigma$.

A rewrite rule is an ordered pair $(l,r)$, written $l\to
r$,  with $l,r\in\Terms$, $l\not\in \Variables$ and
$\Var(r)\subseteq \Var(l)$. The left-hand side ({\em lhs}) of
the rule is $l$ and $r$ is the right-hand side ({\em rhs}). 
A TRS is a pair $\cR=(\Symbols, R)$ where $R$ is a set of rewrite  rules.
$L(\cR)$ denotes the set of $lhs$'s of $\cR$. 
An instance $\sigma(l)$ of a $lhs$ $l$ of a rule is a redex.
The set of redex positions in $t$ is $\Pos_{\cal R}(t)$.
A TRS $\cR$ is left-linear if for all $l\in L(\cR)$, $l$ is a linear
term. 
 Given $\cR=(\Symbols,R)$, we consider $\Symbols$ as the
 disjoint union  $\Symbols=\CSymbols\uplus\DSymbols$ of
 symbols $c\in\CSymbols$, called {\em constructors} and
 symbols $f\in\DSymbols$, called {\em defined functions},
 where $\DSymbols=\{root(l)~|~l\to r\in R\}$
 and $\CSymbols=\Symbols-\DSymbols$. 
 Then, $\CTerms$ (resp. $\GCTerms$) is the set of constructor (resp.\ ground constructor) terms. 
 A TRS $\cR=(\CSymbols\uplus\DSymbols,R)$  is a \emph{constructor system} (CS) if 
 for all $f(\ell_1,\ldots,\ell_k) \to r\in R$, $\ell_i\in\CTerms$, for $1\leq 
 i\leq k$.
 
A term $t\in\Terms$ rewrites to $s$ (at position $p$), written 
$t\stackrel{p}{\to}_\cR s$ (or just $t\to s$), if $t|_p=\sigma(l)$ and 
$s=t[\sigma(r)]_p$, for some rule $\rho:l\rightarrow r\in R$, 
$p\in \Pos(t)$ and substitution $\sigma$. A TRS is terminating if 
$\to$ is terminating. 
 A term $s$ is root-stable (or a head-normal form) if $\forall t$, 
 if $s\to^*t$, then $t$ is not a redex. 
 A term is said to be head-normalizing if it rewrites into a head-normal form.

\section{Context-sensitive rewriting}\label{SecCSR}

A mapping $\mu:\Symbols\rightarrow\wp(\nat)$ is a {\em replacement map} ($\Symbols$-map) if for all
$f\in\Symbols,~\mu(f)\subseteq \{1,\ldots,ar(f)\}$ \cite{Lucas_FundamentalsCSRR_SOFSEM95,Lucas_CScompFuncFunLogProg_JFLP98}.
 $\Rmaps{\Symbols}$ is the set of $\Symbols$-maps. 
Replacement maps can be compared according to their `restriction power':
$\mu\sqsubseteq\mu' \mbox{ if for all } f\in\Symbols,~
\mu(f)\subseteq\mu'(f)$.
If $\mu\sqsubseteq\mu'$, we say that $\mu$ is \emph{more restrictive} than 
$\mu'$. 
Then,
$(\wp(\nat),\subseteq,\emptyset,\nat,\cup)$ induces a complete
lattice $(\Rmaps{\Symbols},\sqsubseteq,\mu_\bot,\mu_\top,\sqcup)$:
the minimum (maximum) element is $\mu_\bot$ ($\mu_\top$), given by
$\mu_\bot(f)=\emptyset$ ($\mu_\top(f)=\{1,\ldots,ar(f)\}$) for all 
$f\in\Symbols$. The {\em lub} $\sqcup$ is 
given by $(\mu\sqcup\mu')(f)=\mu(f)\cup\mu'(f)$ for all $f\in\Symbols$. 

The replacement restrictions introduced by a replacement map $\mu$ 
on the {\em arguments} of function {\em symbols} are raised to {\em positions} of 
{\em terms} $t\in\Terms$: the set $\Pos^\mu(t)$ of {\em $\mu$-replacing positions}  of $t$ is: 
\[\Pos^\mu(t)=\left \{
\begin{array}{ll}
\{\toppos\} & \mbox{if } t\in\Variables\\
\{\toppos\}\cup\:\bigcup_{i\in\mu(root(t))}i.\Pos^\mu(t|_i) & 
\mbox{if }t\not\in\Variables
\end{array}
\right .\]
Given terms $s,t\in\Terms$, $\Pos^\mu_s(t)$ is the set of  
 positions corresponding to $\mu$-replacing occurrences of $s$ in $t$: 
 $\Pos^\mu_s(t)=\Pos^\mu(t)\cap\Pos_s(t)$. 
 The set of  
 {\em $\mu$-replacing 
 variables} occurring in $t\in\Terms$ is 
$\Var^\mu(t)=\{x\in\Variables\mid \Pos^\mu_x(t)\neq\emptyset\}$.

\subsection{Canonical replacement map}

Given $t\in\Terms$, a replacement map $\mu\in\Rmaps{\Symbols}$,
is called \emph{compatible} with $t$ (and vice versa) if 
$\Pos_\Symbols(t)\subseteq\Pos^\mu(t)$. 
Furthermore, $\mu$ is called \emph{strongly compatible} with $t$ if $\Pos_\Symbols(t)=\Pos^\mu(t)$.
And $\mu$ is (strongly) compatible with $T\subseteq\Terms$ if
for all $t\in T$, $\mu$ is (strongly) compatible with $t$ \cite{Lucas_NeedRedCSRR_ALP97,Lucas_CScompFuncFunLogProg_JFLP98}.
The {\em 
minimum} replacement map which is compatible with 
$t\in\Terms$ is \cite{Lucas_CScompFuncFunLogProg_JFLP98}:
\[\mu_t = \left \{ \begin{array}{ll}
\mu_\bot & \mbox{if } t\in\Variables\\
\mu^\toppos_t\sqcup\mu_{t|_1}\sqcup\cdots\sqcup\mu_{t|_{ar(root(t))}} & 
\mbox{if } t\not\in\Variables
\end{array}
\right .\]
with  
$\mu^\toppos_t(root(t))=\{i\in\{1,\ldots,ar(root(t))\}~|~t|_i\not\in\Variables\}$
and $\mu^\toppos_t(f)=\emptyset$ if $f\neq root(t)$.

For a TRS $\cR=(\Symbols,R)$, we use $\Rmaps{\cR}$ instead of 
$\Rmaps{\Symbols}$.
The canonical replacement map $\muCan$ of $\cR$ is 
{\em the most 
restrictive replacement map ensuring that the non-variable subterms of 
the left-hand sides of the rules of $\cR$ are active}.

\begin{definition}{\rm \cite{Lucas_CScompFuncFunLogProg_JFLP98}}
Let $\cR$ be a TRS. The {\em canonical replacement map} of $\cR$ is 
$\muCan=\sqcup_{l\in L(\cR)}\mu_l$.
\end{definition}
Note that $\muCan$ can be automatically associated to $\cR$ by means 
of a very simple calculus:  for each symbol $f\in\Symbols$ and $i\in\{1,\ldots,ar(f)\}$,
$i\in\muCan(f)$  iff  $\exists l\in L(\cR), 
p\in\Pos_\Symbols(l), (root(l|_p)=f\wedge p.i\in\Pos_\Symbols(l)).$
Given a TRS $\cR$,  
$\CRmaps{\cR}=\{\mu\in\Rmaps{\cR}\mid \muCan\sqsubseteq\mu\}$ is the set
of replacement maps that are equal to or \emph{less restrictive} than the canonical 
replacement map. 
If  $\mu\in\CRmaps{\cR}$, we also say that $\mu$ is {\em a} canonical 
replacement map for $\cR$.

\begin{example}\label{ExCanonicalRepMap}
For $\cR$ in Example \ref{Ex6_8_EH11}, we have\footnote{The specification for constant symbols $a$ is omitted, as it is always the empty set $\mu(a)=\emptyset$.}:
\[\begin{array}{rcl@{\hspace{0.5cm}}rcl}
\muCan(\FS)=\muCan(\FL)=\muCan(\Fnats) =\muCan(:)  & =  &\emptyset \\ 
\muCan(+)=\muCan(+_L)=\muCan(\times)=\muCan(\times_L) & = & \{2\}
\end{array}\]
For instance, $\muCan(\FS)=\emptyset$ because for all subterms $\FS(t)$ in the left-hand sides $\ell$ of 
the rules $\ell\to r$ of $\cR$, $t$ is always a \emph{variable}.
However, $\muCan(+)=\{2\}$ because the second argument of $+$ in the 
left-hand side $x+\Fz$ of the first rule in $\cR$ is \emph{not} a variable.

Note that, $\mu$ in Example \ref{Ex6_8_EH11} prescribes $\mu(\FS)=\{1\}$.
Thus, $\muCan\sqsubset\mu$ and  $\mu\in\CRmaps{\cR}$ but $\mu\neq\muCan$.
\end{example}

\subsection{Strongly compatible TRSs}\label{SecStronglyCompatibleTRSs}

Given $t\in\Terms$, the only $\Symbols(t)$-map $\mu$  (if any) which is strongly compatible with $t$
is $\mu_t$ \cite[Proposition 3.6]{Lucas_NeedRedCSRR_ALP97}.
We call $t\in\Terms$ strongly compatible if $\mu_t$ is strongly compatible with $t$.
Similarly, the only $\Symbols(T)$-map $\mu$ which can be strongly compatible with $T$ is $\mu_T=\sqcup_{t\in T}\mu_t$.
We call $T$ \emph{strongly compatible} if $\mu_T$ is strongly compatible with $T$; we call $T$ \emph{weakly
compatible} if $t$ is strongly compatible for all $t\in T$.

\begin{definition}{\rm \cite{Lucas_NeedRedCSRR_ALP97,Lucas_TransfEfEvFP_PLILP97}}\label{DefiStrongRepIndTRS}
A TRS $\cR$ is \emph{strongly} ({\em weakly}) compatible, if $L({\cal R})$ is a strongly (weakly) 
compatible set of terms.
\end{definition}
The only replacement map (if any) which makes $\cR$ strongly compatible is $\muCan$.
For instance, $\cR$ in Example \ref{Ex6_8_EH11} is strongly compatible, but $\mu$ is \emph{not} strongly compatible
with $L(\cR)$ (variable $y$ in the left-hand side of the second rule is $\mu$-replacing).

\subsection{Context-sensitive rewriting}

Given a TRS $\cR=(\Symbols,R)$, $\mu\in\Rmaps{\cR}$, and 
$s,t\in\Terms$, $s$ $\mu$-rewrites to $t$ at position $p$, written $s\csrewpos{p}{\cR,\mu} t$  (or $s\csrew{\cR,\mu}t$, 
$s\csrew{\mu}t$, or even $s\csrew{}t$), if 
$s\stackrel{p}{\to}_\cR t$ and $p\in \Pos^\mu(s)$ \cite{Lucas_FundamentalsCSRR_SOFSEM95,Lucas_CScompFuncFunLogProg_JFLP98}.
A TRS $\cR$ is $\mu$-terminating if $\csrew{\mu}$ is 
terminating. 
Several tools can be used to prove
termination of \csr; for instance,
\AProVE\ and
\muterm, among others.

\begin{remark}
In the following, when considering a TRS $\cR$ together with a canonical replacement map 
$\mu\in\CRmaps{\cR}$, we  often say that $\csrew{\mu}$ performs 
\emph{canonical} context-sensitive rewriting steps \cite{Lucas_CSRewStrat_IC02}.
\end{remark}
The $\csrew{\mu}$-normal forms are 
called \emph{$\mu$-normal forms}, and $\NF^\mu_\cR$ is the set of $\mu$-normal forms for a 
given TRS $\cR$. 
As for unrestricted rewriting, $t\in\NF^\mu_\cR$ if and only if 
$\Pos^\mu_\cR(t)=\emptyset$ (i.e., $t$ contains no $\mu$-replacing redex). 
Rewriting with canonical replacement maps $\mu$ 
has important computational properties that we enumerate here and use below.

\begin{theorem}{\rm \cite[Theorem 8]{Lucas_CScompFuncFunLogProg_JFLP98}} 
\label{TeoFormasMuNormYHeadNormalForm}
Let $\cR$ be a left-linear TRS and $\mu\in\CRmaps{\cR}$. 
Every $\mu$-normal form is a head-normal form.
\end{theorem}

\begin{theorem}{\rm \cite[Theorem 9]{Lucas_CScompFuncFunLogProg_JFLP98}}\label{TheoReescConstructorHeadNormalForm}
Let $\cR=(\Symbols,R)=(\CSymbols\uplus\DSymbols,R)$ be a left-linear 
TRS and $\mu\in\CRmaps{\cR}$. 
Let $s\in\Terms$, and $t=c(t_1,\ldots,t_k)$ for some $c\in\CSymbols$. 
If $s\to^* t$, then there is $u=c(u_1,\ldots,u_k)$
such that $s\hto^*u$ and, for all $i$, $1\leq i\leq k$, $u_i\to^* t_i$.
\end{theorem}

\section{Productivity and termination of \csr}\label{SecInfNormalizationProductivity}

The operational semantics of rewriting-based programming languages can be abstracted, 
for each program (i.e., TRS) $\cR$, as a mapping from terms $s\in\Terms$ into (possibly empty) 
sets of (possibly infinite) terms $T_s\subseteq\ITerms$, which are (possibly infinite) \emph{reducts} of $s$.
The \emph{intended shape} of terms in $T_s$ depends on the application:

\begin{enumerate}
\item In \emph{functional programming}, (ground) \emph{values} $t\in\GCTerms$ are the meaningful
reducts of  (ground) 
initial expressions $s$ (\emph{evaluation} semantics) and $T_s\subseteq\GCTerms$.
\item In \emph{lazy} functional programming \emph{infinite values} are also accepted in the 
semantic description, i.e., $T_s\subseteq\IGCTerms$, but the infinite terms are not actually
obtained but only \emph{approximated} as sequences of appropriate finite terms which are \emph{prefixes}
of the infinite values\footnote{Such finite approximations to infinite terms are described as \emph{partial values}
using a special symbol $\bot$ to denote undefinedness. An infinite value $\delta\in\IGCTerms$ is the limit of an infinite sequence
$\delta_1,\ldots,\delta_n,\ldots$ of such partial values where, for all $i\geq 1$,
$\delta_{i+1}\in\GTermsOn{\CSymbols\cup\{\bot\}}$ is obtained from $\delta_i\in\GTermsOn{\CSymbols\cup\{\bot\}}$ 
by replacing occurrences of $\bot$ in $\delta_i$ by partial values 
different from $\bot$.}.
\item In \emph{equational programming} and rewriting-based theorem provers, computing \emph{normal forms} 
is envisaged (\emph{normalization} semantics), i.e., $T_s\subseteq\NF_\cR$.
\end{enumerate}
 In functional programming (both in the \emph{eager} and \emph{lazy} case), 
 computations can be understood as decomposed into the computation of 
 a head-normal form $t'$ (i.e., $s\to^*t'$) which is then rewritten (below the root!) into $t$.
When a head-normal form $t'$ is obtained, the root symbol $f=root(t')$ is checked. If $f$ is a constructor symbol, then the
evaluation continues on an argument of $t'$. Otherwise, the evaluation \emph{fails} and an error is reported (this
corresponds to $T_s$ empty). Thus, a 
head-normalization process is involved in the computation of the semantic sets $T_s$.

The notion of \emph{productivity} in term rewriting has to do with the ability of TRSs to compute 
possibly infinite \emph{values}.
Most presentations of productivity analysis use sorted signatures and terms
\cite{EndHen_LazyProductivityViaTermination_TCS11,ZanRaf_ProvingProdInInfDataStructures_RTA10}.
The set of sorts $\cS$ is partitioned into $\cS=\Delta\cup\Gamma$, where $\Delta$ is the set of
data sorts, intended to model inductive data types (booleans, natural numbers, finite lists, etc.).
On the other hand, $\Gamma$ is the set of \emph{codata} sorts, intended to model coinductive
datatypes such as streams and infinite trees.
Terms of sort $\Delta$ are called \emph{data terms} and terms of sorts $\Gamma$ are called
\emph{codata} terms.
Given a symbol $f:\tau_1\times\cdots\times\tau_n\to\tau$, $ar_\Delta(f)$ (resp.\ $ar_\Gamma(f)$) 
is the number of arguments of $f$ of sort $\Delta$ (resp.\ $\Gamma$).
Endrullis et al.\ (and also \cite{ZanRaf_ProvingProdInInfDataStructures_RTA10}) assume all
data arguments to be in the first argument positions of the symbols.
\begin{definition}{\rm  \cite[Definition 3.1]{EndHen_LazyProductivityViaTermination_TCS11}}
A \emph{tree specification} is a $(\Delta\cup\Gamma)$-sorted, orthogonal, exhaustive constructor
TRS $\cR$ where $\Delta\cap\Gamma=\emptyset$.
\end{definition}
Here, $\cR$ is called \emph{exhaustive} if for all $f\in\Symbols$, every term $f(t_1,\ldots,t_k)$ is a redex whenever 
$t_i\in\IGCTerms$ are (possiby infinite) closed
constructor terms  for all $i$, $1\leq i\leq k$ \cite[Definition 2.9]{EndHen_LazyProductivityViaTermination_TCS11}.
As in \cite[Definition 2.4]{EndHen_LazyProductivityViaTermination_TCS11}, we assume here
a generalized notion of substitution as an $\cS$-sorted mapping $\sigma:\Variables\to\ITerms$ which is also
extended to a mapping $\sigma:\ITerms\to\ITerms$.

\begin{example}\label{Ex6_8_EH11_treeSpecification}
Consider the tree specification $\cR$ in Example \ref{Ex6_8_EH11},
where, according to \cite[Example 6.8]{EndHen_LazyProductivityViaTermination_TCS11},
$\Delta=\{\mt{Ord}\}$ with $\mt{Ord}$ a data sort for ordinals and
$\Gamma=\{\mt{Str}\}$ with $\mt{Str}$ a codata sort for streams of ordinals.
The types for the constructor symbols are: $\Fz::\mt{Ord}$, $\FS::\mt{Ord}\to\mt{Ord}$, $\FL::\mt{Str}\to\mt{Ord}$
and $(:)::\mt{Ord}\times\mt{Str}\to\mt{Str}$.
Thus, 
$\CSymbols_\Delta=\{\Fz,\FS,\FL\}$, 
$\CSymbols_\Gamma=\{:\}$, 
$\DSymbols_\Delta=\{+,\times,\omega\}$, and
$\DSymbols_\Gamma=\{+_L,\times_L,\Fnats\}$. 
\end{example}

\begin{definition}{\rm \cite[Definition 3.5]{EndHen_LazyProductivityViaTermination_TCS11}}
A tree specification $\cR$ is \emph{constructor normalizing} if all finite ground terms $t\in\GTerms$
rewrite to a possibly infinite constructor normal form $\delta\in\IGCTerms$.
\end{definition}
Being exhaustive is a necessary condition for productivity.

\begin{theorem}
If $\cR$ is constructor normalizing, then it is exhaustive.
\end{theorem}

\begin{proof}
If not, then there is a finite ground normal form $t$ containing a defined symbol. This contradicts $\cR$
being constructor normalizing.
\end{proof}
\begin{theorem}\label{TheoCanonicalMuTermAndConstructorNormalizationOfTreeSpec}
Let $\cR$ be an exhaustive, left-linear TRS and $\mu\in\CRmaps{\cR}$. 
If $\cR$ is $\mu$-terminating, then $\cR$ is constructor normalizing.
\end{theorem}

\begin{proof}
Since $\cR$ is $\mu$-terminating, every ground term $s$ has a (finite) $\mu$-normal form $t$. 
By Theorem \ref{TeoFormasMuNormYHeadNormalForm}, $t$ is a head-normal form.
We prove by induction on $t$ that $t$ rewrites into a (possibly infinite) constructor term $\delta\in\IGCTerms$.
If $t$ is a constant, then since $t$ is a $\mu$-normal form, it must be a normal form.
Since $\cR$ is exhaustive, $t=\delta\in\GCTerms$.
If $t=f(t_1,\ldots,t_k)$ for ground terms $t_1,\ldots,t_k$, 
then by the induction hypothesis, for all $i$, $1\leq i\leq k$, 
$t_i$ has a (possibly infinite) constructor normal form $\delta_i\in\IGCTerms$. 
We have two cases:
\begin{enumerate}
\item If $f\in\CSymbols$, then $t$ has a (possibly infinite) constructor normal form $f(\delta_1,\ldots,\delta_k)$.
\item If $f\notin\CSymbols$, then, since $t$ is a head-normal form, 
$f(\delta_1,\ldots,\delta_k)$ is a ground 
(possibly infinite) normal form which contradicts that $\cR$ is exhaustive.
\end{enumerate}
Thus, $s$ has a (possibly infinite) constructor normal form as well and $\cR$ is constructor normalizing.
\end{proof}
Since tree specifications are left-linear and exhaustive, Theorem \ref{TheoCanonicalMuTermAndConstructorNormalizationOfTreeSpec} holds for tree specifications.

\begin{example}\label{Ex4_6_ZR10}
The following tree specification $\cR$ (cf.\ \cite[Example 4.6]{ZanRaf_ProvingProdInInfDataStructures_RTA10})
\begin{eqnarray}
\Fp & \to & \Fzip(\Falt,\Fp)\nonumber\\
\Falt & \to & \Fz:\Fone:\Falt\nonumber\\
\Fzip(x:\sigma,\tau) & \to & x:\Fzip(\tau,\sigma)\nonumber
\end{eqnarray}
(where no constant for empty lists is included!) is easily proved $\muCan$-terminating (use \muterm). By Theorem \ref{TheoCanonicalMuTermAndConstructorNormalizationOfTreeSpec}, it is constructor normalizing.
Note that $\cR$ is \emph{exhaustive} due to the sort discipline (for instance, $zip(0,0)$ is not allowed) 
and to the fact that no constructor for lists is provided (i.e., there is no finite list and all lists are of the form $cons(s,t)$
for terms $s$, $t$ where $t$ is always infinite).
\end{example}

As remarked in \cite[Section 3.2]{EndHen_LazyProductivityViaTermination_TCS11}, several authors
define $\cR$ to be productive if it is constructor normalizing 
(e.g., \cite{
EndGraHenIsiKlo_ProductivityOfStreamDefinitions_TCS10,%
ZanRaf_StreamProductivityOutermostTermination_EPTCS09,%
ZanRaf_ProvingProdInInfDataStructures_RTA10}).
Endrullis and Hendriks give a more elaborated (and restrictive) definition of productivity.
Given $t\in\ITerms$ and $\Symbols'\subseteq\Symbols$, a $\Symbols'$-path in $t$ is a (finite or infinite) sequence
$\langle p_1,c_1\rangle$, $\langle p_2,c_2\rangle, \ldots$ such that $c_i=root(t|_{p_i})\in\Symbols'$ and 
$p_{i+1}=p_i.j$ with $1\leq j\leq ar(c_i)$  \cite[Definition 3.7]{EndHen_LazyProductivityViaTermination_TCS11}.
\begin{definition}{\rm  \cite[Definition 3.8]{EndHen_LazyProductivityViaTermination_TCS11}}
A tree specification is said \emph{data-finite} if for all finite ground terms $s\in\GTerms$ and (possibly infinite) 
constructor normal forms $t$ of $s$, every $\CSymbols_\Delta$-path in $t$ (containing data constructors only) is
finite.
\end{definition}

\begin{definition}{\rm  \cite[Definition 3.11]{EndHen_LazyProductivityViaTermination_TCS11}}
A tree specification $\cR$ is \emph{productive} if $\cR$ is constructor normalizing and data-finite.
\end{definition}
In the following result, $\mu_\Delta$ is given by $\mu_\Delta(c)=\{1,\ldots,ar_\Delta(c)\}$ for all  $c\in\CSymbols_\Delta$, 
and $\mu_\Delta(f)=\emptyset$
for all other symbols $f$.

\begin{theorem}\label{TheoCanonicalMuTermAndProductivityOfTreeSpec}
Let $\cR$ be a left-linear, exhaustive TRS and $\mu\in\Rmaps{\cR}$ be such that
$\muCan\sqcup\mu_\Delta\sqsubseteq\mu$. If $\cR$ is $\mu$-terminating, then $\cR$ is productive.
\end{theorem}

\begin{proof}
Since $\muCan\sqsubseteq\mu$, constructor normalization of $\cR$ follows by Theorem \ref{TheoCanonicalMuTermAndConstructorNormalizationOfTreeSpec}.
Thus, if $\cR$ is not productive, there must be a ground normal form $t$ of a term $s$ with an infinite
$\CSymbols_\Delta$-path.
Without loss of generality, we can assume that 
$s\to^*s_1=c_1(s^1_1,\ldots,s^1_{k_1})$ for some $c_1\in\CSymbols_\Delta$,
and then 
$s^1_{i_1}\to^*c_2(s^2_1,\ldots,s^2_{k_2})$ for some $i_1$, $1\leq	i_1\leq ar_\Delta(c_1)$ and $c_2\in\CSymbols_\Delta$, etc., in such a way that this reduction sequences follow the computation of $t$ and produce the
$\CSymbols_\Delta$-path $\langle\toppos,c_1\rangle$, $\langle i_1,c_2\rangle$, $\langle i_1.i_2,c_3\rangle$,$\ldots$

By Theorem \ref{TheoReescConstructorHeadNormalForm}, 
$s\hto^*\ol{s}_1=c_1(\ol{s}^1_1,\ldots,\ol{s}^1_{k_1})$ for some terms 
$\ol{s}^1_1,\ldots,\ol{s}^1_{k_1}$ such that $\ol{s}^1_j\to^*s^1_j$ for all $j$, $1\leq j\leq k_1$.
Thus, by Theorem \ref{TheoReescConstructorHeadNormalForm} we also 
have $\ol{s}^1_{i_1}\hto^*c_2(\ol{s}^2_1,\ldots,\ol{s}^2_{k_2})$ and $\ol{s}^2_j\to^*s^2_j$ for all $j$, $1\leq j\leq k_2$.
Since $i_1\in\mu_\Delta(c_1)$, we have $s\hto^*\ol{s}_2=c_1(\ol{s}^1_1,\ldots,\ol{s}^1_{i_1-1},c_2(\ol{s}^2_1,\ldots,\ol{s}^2_{i_2},\ldots,\ol{s}^2_{k_2}),\ldots,\ol{s}^1_{k_1})$ with $\ol{s}^2_{i_2}\to^*s^2_{i_2}$ again.
Since $i_1.i_2\in\Pos^{\mu_\Delta}(\ol{s}_2)$, we can continue with this construction to obtain an infinite $\mu$-rewriting
sequence which contradicts $\mu$-termination of $\cR$.
\end{proof}

\begin{example}\label{Ex6_8_EH11_Productivity}
For the tree specification $\cR$ in Example \ref{Ex6_8_EH11} (see also
Example \ref{Ex6_8_EH11_treeSpecification}), we have
 $ar_\Delta(\FS)=1$ and 
$ar_\Delta(\FL)=0$. 
Then, $\mu_\Delta(\FS)=\{1\}$ and $\mu_\Delta(\FL)=\emptyset$. 
Now $\mu=\muCan\sqcup\mu_\Delta$ is as given in Example \ref{Ex6_8_EH11}.
The $\mu$-termination of $\cR$ can be proved with $\muterm$.
By Theorem \ref{TheoCanonicalMuTermAndProductivityOfTreeSpec}, productivity of $\cR$ follows.
\end{example}

\begin{example}\label{Ex4_6_ZR10_Productivity}
We also prove productivity of $\cR$ in Example \ref{Ex4_6_ZR10}.
Here, $\Delta=\{d\}$ and $\Gamma=\{s\}$ with $\CSymbols_\Delta=\{\Fz,\Fone\}$ and $\CSymbols_\Gamma=\{\Fcons\}$
where $ar_\Delta(\Fcons)=1$.
Thus, $\mu=\muCan\sqcup\mu_\Delta$ yields $\mu(\Fzip)=\mu(\Fcons)=\{1\}$.
The $\mu$-termination of $\cR$ can be proved with $\muterm$ and by Theorem \ref{TheoCanonicalMuTermAndProductivityOfTreeSpec}  productivity of $\cR$ follows.
\end{example}
In general, Theorem \ref{TheoCanonicalMuTermAndProductivityOfTreeSpec} does \emph{not} hold in the opposite direction, i.e., productivity of $\cR$ does not imply
its $\mu$-termination. 

\begin{example}\label{Ex5_3_EH11}
Let $\cR$ be (cf.\ \cite[Example 5.3]{EndHen_LazyProductivityViaTermination_TCS11}):
\begin{eqnarray}
\Fs & \to & \Fb:\Fs\nonumber\\
\Ff(\Fa,\sigma) & \to & \sigma\nonumber\\
\Ff(\Fb,x:y:\sigma) & \to & \Fb:\Ff(\Fb,y:\sigma)\nonumber
\end{eqnarray}
Note that
$\muCan(:)=\{2\}$ due to the third rule. 
This makes $\cR$ non-$\muCan$-terminating due to the first rule.
We cannot use Theorem \ref{TheoCanonicalMuTermAndProductivityOfTreeSpec}
to prove $\cR$ productive, but it is (see Example \ref{Ex5_3_EH11_Transformed} below).
\end{example}
Regarding constructor normalization, we have:

\begin{theorem}\label{TheoConstructorNormalizationImpliesMuCanTerm}
Let $\cR$ be a 
orthogonal strongly compatible TRS
such that either
\begin{enumerate}
\item $\muCan(c)=\emptyset$ for all $c\in\CSymbols$, or
\item $\cR$ contains no collapsing rule and $\muCan(c)=\emptyset$ for all 
constructor symbols $c\in\CSymbols_\cR$ such that $c=root(r)$ for some $\ell\to r\in\cR$.
\end{enumerate}
If $\cR$ is constructor normalizing, then it is $\muCan$-terminating.
\end{theorem}

\begin{proof}
Since $\cR$ is constructor normalizing, $\cR$ is head-normalizing, i.e., every term $s$ has a (constructor) 
head-normal form $t$, i.e., $root(t)\in\CSymbols$. 
By \cite[Theorem 4.6]{Lucas_NeedRedCSRR_ALP97}, 
every $\muCan$-replacing redex in a term $s$ which is not a head-normal form is root-needed
(see \cite{Midd_CBNCompRootStableForm_POPL97}). Thus, 
every $\muCan$-reduction sequence with $\cR$ is head-normalizing. 
Furthermore, since every term $s$ is head-normalizing, 
every $\muCan$-rewrite sequence starting from $s$
yields a head-normal form $t$ which, by confluence of $\cR$, is a constructor head-normal form, i.e.,
$t=c(t_1,\ldots,t_k)$ for some $c\in\CSymbols$.
We have two cases:
\begin{enumerate}
\item If $\muCan(c)=\emptyset$ for all constructor symbols $c$, then
$t$ is a $\mu$-normal form.
\item Otherwise, we can assume that $s$ is not  a head-normal form and then, since there is no
collapsing rule, the root symbol $c$ of $t$ must be introduced by the last rule applied to the root
in the head-normalizing sequence. 
Hence,  by our assumption, $\mu(c)=\emptyset$ as well. 
\end{enumerate}
Thus, every $\muCan$-rewrite sequence starting from any term $s$ is finite and 
$\cR$ is $\muCan$-terminating.
\end{proof}

\section{Related work}\label{SecRelatedWork}
In~\cite{ZanRaf_ProvingProdInInfDataStructures_RTA10}, Zantema and Raffelsieper
develop a general technique to prove productivity of specifications of
infinite objects based on proving context-sensitive termination.
In the following result, we use the terminology in Section \ref{SecInfNormalizationProductivity}, 
borrowed from \cite{EndHen_LazyProductivityViaTermination_TCS11}.
Consistently, since the notion of `productivity' in \cite{ZanRaf_ProvingProdInInfDataStructures_RTA10},
corresponds to constructor normalization (see Section  \ref{SecInfNormalizationProductivity}),
we have the following.

\begin{theorem}{\rm \cite[Theorem 4.1]{ZanRaf_ProvingProdInInfDataStructures_RTA10}}\label{TheoZR10_Theorem4_1}
Let $\cR$ be a \emph{proper} tree specification and $\mu\in\Rmaps{\cR}$ given by $\mu(f)=\{1,\ldots,ar(f)\}$
if $f\in\DSymbols$ and $\mu(c)=\{1,\ldots,ar_\Delta(c)\}$ if $c\in\CSymbols$.
If $\cR$ is $\mu$-terminating, then $\cR$ is constructor normalizing.
\end{theorem}
\begin{remark}
Theorem \ref{TheoZR10_Theorem4_1} is a particular case of Theorem \ref{TheoCanonicalMuTermAndConstructorNormalizationOfTreeSpec}: proper tree specifications
are TRSs with rules $\ell\to r$ whose left-hand sides $\ell$ contain no nested constructor symbols,
i.e., they are of the form $\ell=f(\delta_1,\ldots,\delta_k)$, where $\delta_i$ is either a variable or a \emph{flat} constructor
term $c_i(x_1,\ldots,x_m)$ for some constructor symbol $c_i$ and variables $x_1,\ldots,x_m$.
In this case, the replacement map $\mu$ required in Theorem \ref{TheoZR10_Theorem4_1} is \emph{canonical}, i.e.,
$\mu\in\CRmaps{\cR}$.
\end{remark}
Example \ref{Ex4_6_ZR10} is given in \cite[Example 4.6]{ZanRaf_ProvingProdInInfDataStructures_RTA10}
to illustrate a tree specification $\cR$ where Theorem \ref{TheoZR10_Theorem4_1} 
can \emph{not} be used to prove constructor normalization.
Indeed, $\cR$ is not $\mu$-terminating if $\mu$ is defined as required in Theorem \ref{TheoZR10_Theorem4_1}.
In contrast, 
Theorem \ref{TheoCanonicalMuTermAndConstructorNormalizationOfTreeSpec} was used
in Example \ref{Ex4_6_ZR10} to prove constructor normalization of $\cR$
and Theorem \ref{TheoCanonicalMuTermAndProductivityOfTreeSpec} 
was used in Example \ref{Ex4_6_ZR10_Productivity} to prove productivity of $\cR$.

In \cite{EndHen_LazyProductivityViaTermination_TCS11} Endrullis and Hendriks 
have devised a sound and complete transformation of productivity
to context-sensitive termination.
The transformation proceeds in two steps. First, an \emph{inductively sequential} (see \cite{Antoy_DefTrees_ALP92}) 
tree specification $\cR$ is
transformed into a \emph{shallow} tree specification $\cR'$ by a \emph{productivity preserving} transformation
\cite[Definition 5.1]{EndHen_LazyProductivityViaTermination_TCS11} and 
\cite[Theorem 5.5]{EndHen_LazyProductivityViaTermination_TCS11}.
Here,
$\cR$ is \emph{shallow} if for each $k$-ary defined symbol $f\in\DSymbols$ there is a set $I_f\subseteq\{1,\ldots,k\}$ 
such that for each rule $f(p_1,\ldots,p_k)\to r$, every $p_i$ satisfies   \cite[Definition 3.14]{EndHen_LazyProductivityViaTermination_TCS11}:
\begin{enumerate}
\item If $i\in I_f$, then $p_i=c_i(x_1,\ldots,x_m)$ for some $c\in\CSymbols$ and variables $x_1,\ldots,x_m\in\Variables$; and
\item If $i\notin I_f$, then $p_i\in\Variables$.
\end{enumerate}
\begin{example}\label{Ex5_3_EH11_Transformed}
The (inductively sequential) TRS $\cR$ in Example \ref{Ex5_3_EH11} is \emph{not} shallow, but it is transformed 
by the first transformation into the following TRS $\cR'$ (adapted from \cite[Example 5.3]{EndHen_LazyProductivityViaTermination_TCS11}):
\begin{eqnarray}
\Fs & \to & \Fb:\Fs\nonumber\\
\Ff(\Fa,\sigma) & \to & \Ffa(\sigma)\nonumber\\
\Ffa(\sigma) & \to & \sigma\nonumber\\
\Ff(\Fb,\sigma) & \to & \Ffb(\sigma)\nonumber\\
\Ffb(x:\sigma) & \to & \Ffbcons(x,\sigma)\nonumber\\
\Ffbcons(x,y:\sigma) & \to & \Fb:\Ff(\Fb,y:\sigma)\nonumber
\end{eqnarray}
Since $\cR'$ is productive if and only if $\cR$ is, we use 
now Theorem \ref{TheoCanonicalMuTermAndProductivityOfTreeSpec} (with $\mu=\muCan$, since $\mu_\Delta=\mu_\bot$) 
to prove $\cR$ productive.
This shows (see Example \ref{Ex5_3_EH11}) that Theorem \ref{TheoCanonicalMuTermAndProductivityOfTreeSpec} 
 does not extend to a characterization of 
productivity as termination of \csr.
\end{example}

\begin{proposition}\label{PropShallowTreeSpecAreStronglyCompatible}
Shallow tree specifications $\cR$ are strongly compatible constructor TRSs where $\muCan(c)=\emptyset$ for all
$c\in\CSymbols$.
\end{proposition}

\begin{proof}
Let $\mu(f)=I_f$ for all $f\in\DSymbols$ and $\mu(f)=\emptyset$ for all $f\in\CSymbols$.
For all $\ell\in L(\cR)$, $\Pos^\mu(\ell)=\Pos_\Symbols(\ell)$, i.e., $\cR$ is strongly
compatible. Since $\muCan$ is the only replacement map that makes $\cR$ strongly compatible,  
$\mu=\muCan$ and $\muCan(c)=\emptyset$ for all $c\in\CSymbols$.
\end{proof}
In Endrullis and Hendriks' approach, a second transformation obtains a CS-TRS $(\cR'',\mu)$  from $\cR'$ 
 (see \cite[Definition 6.1]{EndHen_LazyProductivityViaTermination_TCS11}) in such a way that 
 $\mu$-termination of $\cR''$ is \emph{equivalent}
to productivity of $\cR'$ \cite[Theorem 6.6]{EndHen_LazyProductivityViaTermination_TCS11}.

\begin{remark}
First Endrullis and Hendriks' transformation preserves productivity. Thus, we can use
$\cR'$ together with Theorem \ref{TheoCanonicalMuTermAndProductivityOfTreeSpec} to prove productivity
of $\cR$ \emph{without using} the second transformation. We proceed in this way
in Example \ref{Ex5_3_EH11_Transformed}, where we conclude 
productivity of $\cR'$ without using the second transformation described in 
\cite[Definition 6.1]{EndHen_LazyProductivityViaTermination_TCS11}.
\end{remark}
By Theorem \ref{TheoConstructorNormalizationImpliesMuCanTerm} and Proposition \ref{PropShallowTreeSpecAreStronglyCompatible}, we have:

\begin{corollary}
Constructor normalizing 
shallow tree specifications 
$\cR$ 
are $\muCan$-terminating.
\end{corollary}

With Theorem \ref{TheoCanonicalMuTermAndConstructorNormalizationOfTreeSpec}, 
we have the following characterization of shallow 
tree specifications (see also \cite[Theorem 6.5]{EndHen_LazyProductivityViaTermination_TCS11}).
\begin{corollary}\label{CoroCanonicalMuTermAndConstructorNormalizationOfShallowTreeSpec}
A shallow tree specification  $\cR$ is constructor normalizing if and only if it is $\muCan$-terminating.
\end{corollary}
However, we also have
\begin{corollary}\label{CoroCanonicalMuTermAndConstructorNormalizationOfStronglyCompatibleTreeSpec}
A strongly compatible tree specification $\cR$ without collapsing rules and such that
$\muCan(c)=\emptyset$ for all 
constructor symbols $c\in\CSymbols_\cR$ such that $c=root(r)$ for some $\ell\to r\in\cR$ 
is constructor normalizing if and only if it is $\muCan$-terminating.
\end{corollary}
Since productive tree specifications are constructor normalizing, we have the following.

\begin{corollary}
Productive shallow tree specifications $\cR$ are $\muCan$-terminating.
\end{corollary}

In \cite{Raffelsieper_ProductivityOfNonOrthogonalTRSs_EPTCS12}, Raffelsieper investigates productivity
of non-orthogonal TRSs.
However, he still requires left-linearity and exhaustiveness of $\cR$. Thus, our results in 
Section \ref{SecInfNormalizationProductivity} also apply to his framework.
Raffelsieper also introduces the notion of \emph{strong productivity} meaning that every maximal outermost-fair
$\cR$-sequence starting from a term of sort $\Delta$ is constructor head-normalizing 
\cite[Definition 6 and Proposition 7]{Raffelsieper_ProductivityOfNonOrthogonalTRSs_EPTCS12}.
He also uses termination of \csr\ to prove strong productivity of his \emph{proper specifications}.
He defines a replacement map $\mu_\cS$ (see \cite[Definition 11]{Raffelsieper_ProductivityOfNonOrthogonalTRSs_EPTCS12})
which is, however, \emph{less} restrictive than our replacement map
$\mu_\Delta$ in Theorem \ref{TheoCanonicalMuTermAndProductivityOfTreeSpec}.
Thus, his main result in this respect \cite[Theorem 12]{Raffelsieper_ProductivityOfNonOrthogonalTRSs_EPTCS12}
is a particular case of our Theorem \ref{TheoCanonicalMuTermAndProductivityOfTreeSpec}.

\section{Conclusions and future work}\label{SecConclusions}

We have identified Theorems \ref{TeoFormasMuNormYHeadNormalForm} and \ref{TheoReescConstructorHeadNormalForm} 
(originally in \cite{Lucas_CScompFuncFunLogProg_JFLP98}) as bearing the essentials of the use of termination of
\emph{canonical} \csr\ to prove productivity of rewrite systems (see the proofs of Theorems \ref{TheoCanonicalMuTermAndConstructorNormalizationOfTreeSpec} and  \ref{TheoCanonicalMuTermAndProductivityOfTreeSpec}).
Although termination of \csr\ had been used before to prove (and even characterize) productivity,
we believe that our presentation sheds new light on this connection and also shows that the use of
such well-known results about \csr\ also simplifies the proofs of the results that connect termination of
\csr\ and productivity.
Furthermore, the use of the canonical replacement map as one of the (bounding) components of the
replacement map at stake is new in the literature and improves on previous approaches that systematically 
use less restrictive replacement maps, thus losing opportunities to prove termination of \csr\ and hence 
productivity.
We improved Endrullis and Hendriks' approach because we avoid the use of transformations,
being able to directly prove productivity of a non-shallow TRS $\cR$ as termination of \csr\ for $\cR$ itself.
For instance, we directly prove productivity of $\cR$ in Example \ref{Ex6_8_EH11} without any transformation, whereas
Endrullis and Hendriks require the addition of new rules due to their second transformation (see
\cite[Example 6.8]{EndHen_LazyProductivityViaTermination_TCS11}).
In Example \ref{Ex5_3_EH11_Transformed}, we conclude 
productivity of $\cR'$ without using their second transformation.
As a matter of fact, we were able to find automatic proofs of productivity for all
the examples  in \cite{EndHen_LazyProductivityViaTermination_TCS11,%
Raffelsieper_ProductivityOfNonOrthogonalTRSs_EPTCS12,%
ZanRaf_ProvingProdInInfDataStructures_RTA10} by using Theorem 
\ref{TheoCanonicalMuTermAndProductivityOfTreeSpec}
together with \AProVE\ or \muterm\ to obtain the automatic proofs of termination of \csr.
Our results, though, do \emph{not} provide a characterization of productivity, as witnessed by Examples
\ref{Ex5_3_EH11} and \ref{Ex5_3_EH11_Transformed}.
In contrast to \cite{EndHen_LazyProductivityViaTermination_TCS11,ZanRaf_ProvingProdInInfDataStructures_RTA10},
which deal with \emph{orthogonal} (constructor-based) TRSs only, our results apply to \emph{left-linear} TRSs and 
supersede \cite{Raffelsieper_ProductivityOfNonOrthogonalTRSs_EPTCS12}
which applies to non-orthogonal TRSs which are still left-linear.

In the future, we plan to apply other powerful results about completeness of \csr\ in (infinitary) normalization and
computation of (possibly infinite) values to develop more general notions of productivity and apply them to 
broader classes of programs.

\paragraph{Acknowledgments.} I thank the anonymous referees for their comments and suggestions.

\bibliographystyle{eptcs}
\bibliography{biblio}

\end{document}